\documentclass[final,3p]{elsarticle}

\usepackage[OT1]{fontenc}
\usepackage[english]{babel}
\usepackage[utf8]{inputenc}
\usepackage[svgnames]{xcolor}
\usepackage{amssymb,amsfonts,mathrsfs,yhmath}
\usepackage[amsthm,amsmath,thmmarks]{ntheorem}
\usepackage{enumitem}
\usepackage{expdlist,graphicx}
\usepackage{booktabs}
\usepackage{caption}
\usepackage[labelformat=simple]{subcaption}

\usepackage{epsfig}
\usepackage[numbers]{natbib}
\usepackage[linkcolor=Blue,colorlinks=true,citecolor=ForestGreen,filecolor=red,urlcolor=Blue]{hyperref}
\usepackage{cleveref}[2011/12/24]
\usepackage[pagewise,displaymath,mathlines]{lineno}

\theoremstyle{plain}
\newtheorem{theorem}{Theorem}
\newtheorem{corollary}[theorem]{Corollary}
\newtheorem{lemma}[theorem]{Lemma}
\newtheorem{proposition}[theorem]{Proposition}
\newtheorem{fact}[theorem]{Fact}

\newtheorem{claim}[theorem]{Claim}
\theoremstyle{definition}
\newtheorem{definition}[theorem]{Definition}

\theoremstyle{remark}

\date{}

\begin{document}
\begin{frontmatter}

\title{The VC-Dimension of Graphs with Respect to $k$-Connected Subgraphs}
\author{Andrea Munaro\corref{f1}} 
\ead{Andrea.Munaro@grenoble-inp.fr}
\address{Laboratoire G-SCOP, Universit\'{e} Grenoble Alpes}
\cortext[f1]{This work was partly done at Dept. of Computer Science, University of Bonn.}

\begin{abstract} 
We study the VC-dimension of the set system on the vertex set of some graph which is induced by the family of its $k$-connected subgraphs. In particular, we give tight upper and lower bounds for the VC-dimension. Moreover, we show that computing the VC-dimension is $\mathsf{NP}$-complete and that it remains $\mathsf{NP}$-complete for split graphs and for some subclasses of planar bipartite graphs in the cases $k = 1$ and $k = 2$. On the positive side, we observe it can be decided in linear time for graphs of bounded clique-width. 
\end{abstract}


\end{frontmatter}

\section{Introduction}

The parameter now called VC-dimension of a set system was introduced by \citet{VC71}. The initial interest was in the contexts of empirical process theory and learning theory, where it proved to be a fundamental concept. It represents a prominent measure of the ``complexity'' of the set system. Let $\mathcal{H}$ be a set system on a finite set $X$. A subset $Y \subseteq X$ is \textit{shattered by $\mathcal{H}$} if $\left\{E \cap Y : E \in \mathcal{H}\right\} = 2^{Y}$. The \textit{VC-dimension} of $\mathcal{H}$ is defined as the maximum size of a set shattered by $\mathcal{H}$. One might think to apply the abstract notion of VC-dimension to some concrete settings. A natural choice is the study of the VC-dimension associated to graphs. Given a graph, we consider set systems induced by a certain family of subgraphs. In this way we obtain several different notions of VC-dimension, each one related to a special family of subgraphs. This study was first initiated in a seminal paper by \citet{HW87}. They considered the set system induced by closed neighbourhoods of the vertices. \citet{KKRUW97} investigated the VC-dimensions induced by other families of subgraphs. They adapted the definition of VC-dimension to the graph theoretic setting as follows.

\begin{definition}\label{classes} Let $G = (V, E)$ be a graph and let $\mathcal{P}$ be a family of subgraphs of $G$. A subset $A \subseteq V$ is \textit{$\mathcal{P}$-shattered} if every subset of $A$ can be obtained as the intersection of $V(H)$, for $H \in \mathcal{P}$, with $A$. The \textit{VC-dimension} of $G$ with respect to $\mathcal{P}$, denoted by $\textsc{VC}_{\mathcal{P}}(G)$, is defined as the maximum size of a $\mathcal{P}$-shattered subset.   
\end{definition}

Note that in this paper we consider only finite undirected simple graphs and we use standard graph theoretic terminology from \cite{West}, unless stated otherwise. For definitions and examples related to tree-width and clique-width we refer the reader to \citep{CO00, KLM09}, while for an introduction to monadic second-order logic of graphs see, for example, \citep{Cou90}.

According to \Cref{classes}, we denote by $\textsc{VC}_{\mbox{\scriptsize{\normalfont{tree}}}}$, $\textsc{VC}_{\mbox{\scriptsize{\normalfont{con}}}}$, $\textsc{VC}_{k-\mbox{\scriptsize{\normalfont{con}}}}$, $\textsc{VC}_{\mbox{\scriptsize{\normalfont{nbd}}}}$, $\textsc{VC}_{\mbox{\scriptsize{\normalfont{path}}}}$, $\textsc{VC}_{\mbox{\scriptsize{\normalfont{cycle}}}}$ and $\textsc{VC}_{\mbox{\scriptsize{\normalfont{star}}}}$ the VC-dimensions with respect to families of subgraphs which are tree, connected, $k$-connected, closed neighbourhood, path, cycle and star, respectively. Note that the VC-dimension with respect to some families of subgraphs is equal to well established quantities in graph theory: if $\mathcal{P}$ is the family of complete subgraphs then $\textsc{VC}_{\mathcal{P}}$ is the clique number, while if $\mathcal{P}$ is the family of subgraphs induced by independent sets then $\textsc{VC}_{\mathcal{P}}$ is the independence number.

Since a graph of order $n$ has $n$ closed neighbourhoods, then its VC-dimension is at most $\left\lfloor \log_{2}n \right\rfloor$ \cite{HW87}. It is not difficult to show that this bound is tight \cite{ABC95}. Indeed, consider the graph $H$ built as follows. Take a set $S$ of $\left\lfloor\log_{2}n\right\rfloor$ independent vertices. For each non-singleton subset $R \subseteq S$, add a vertex $v_{R}$ adjacent to precisely the vertices of $R$. The resulting graph $H$ has at most $n$ vertices and $\textsc{VC}_{\mbox{\scriptsize{\normalfont{nbd}}}}(H) = \left\lfloor\log_{2}n\right\rfloor$. If $G$ is a graph with maximum degree $\Delta$, then it is easy to see that $\Delta \leq \textsc{VC}_{\mbox{\scriptsize{\normalfont{star}}}} \leq \Delta + 1$ \cite{KKRUW97}. The VC-dimension with respect to trees is the same as the VC-dimension with respect to connected subgraphs \cite{KKRUW97}. This is an immediate consequence of the fact that a connected graph contains a spanning tree. 

\citet{KKRUW97} related the VC-dimension of a graph $G$ with respect to connected subgraphs to the number of leaves $\ell(G)$ in a maximum leaf spanning tree of $G$.  

\begin{theorem}[\citet{KKRUW97}]\label{con} $\ell(G) \leq \textsc{VC}_{\mbox{\scriptsize{\normalfont{con}}}}(G) \leq \ell(G) + 1$, for any graph $G$. 
\end{theorem}

Recall that $\ell(G) = |V(G)| - \gamma_{c}(G)$, where $\gamma_{c}(G)$ is the connected domination number, the minimum size of a connected dominating set in $G$. Moreover, deciding $\gamma_{c}(G)$ (and therefore $\ell(G)$) is $\mathsf{NP}$-complete \citep{GJ79}. A natural question is to investigate the computational complexity of computing $\textsc{VC}_{\mathcal{P}}(G)$ for a given graph $G$ and a family of its subgraphs $\mathcal{P}$. The decision problem is formulated as follows:

\begin{center}
\fbox{%
\begin{minipage}{5.5in}
\textsc{Graph $\textsc{VC}_{\mathcal{P}}$ Dimension}
\begin{description}[\compact\breaklabel\setleftmargin{70pt}]
\item[Instance:] A graph $G$ and a number $s \geq 1$.
\item[Question:] Does $\textsc{VC}_{\mathcal{P}}(G) \geq s$ hold?  
\end{description}
\end{minipage}}
\end{center}

\subsection{Our results}

In \Cref{second} we extend \Cref{con} by giving tight upper and lower bounds on the VC-dimension with respect to $k$-connected subgraphs, for $k \geq 2$. These are given, similarly to \Cref{con}, in terms of the number of leaves in a maximum leaf spanning tree. In \Cref{third} we prove that the related decision problem \textsc{Graph $\textsc{VC}_{k-\mbox{\scriptsize{\normalfont{con}}}}$ Dimension} is $\mathsf{NP}$-complete even for split graphs. On the positive side, we show it can be decided in linear time for graphs of bounded clique-width and in polynomial time for for the subclass of split graphs having Dilworth number at most $2$. Finally, we prove that \textsc{Graph $\textsc{VC}_{\mbox{\scriptsize{\normalfont{con}}}}$ Dimension} and \textsc{Graph $\textsc{VC}_{2-\mbox{\scriptsize{\normalfont{con}}}}$ Dimension} remain $\mathsf{NP}$-complete for some subclasses of planar bipartite graphs with maximum degree at most $4$. The following table summarizes the known results about the complexity of \textsc{Graph $\textsc{VC}_{\mathcal{P}}$ Dimension}. 

\begin{center}
\begin{tabular}{@{}lllr@{}} \toprule
Family $\mathcal{P}$ & Graph $G$ & Computational Complexity & Reference \\ \midrule
star &  & polynomial time & \citet{KKRUW97} \\
neighbourhood &  & $\mathsf{LOGNP}$-complete & \citet{KKRUW97} \\
path &  & $\mathsf{\Sigma}^{p}_{3}$-complete & \citet{Sch00} \\
cycle &  & $\mathsf{\Sigma}^{p}_{3}$-complete & \citet{Sch00} \\
$k$-connected & split & $\mathsf{NP}$-complete & \Cref{reduction} \\
$k$-connected & bounded clique-width & linear time & \Cref{treeclique} \\
$k$-connected & split, Dilworth number $\leq 2$ & polynomial time & \Cref{thres} \\
connected & planar, bipartite, $\Delta(G) = 3$ & $\mathsf{NP}$-complete & \Cref{redpla} \\
$2$-connected & planar, bipartite, $\Delta(G) = 4$ & $\mathsf{NP}$-complete & \Cref{redpla2con} \\ \bottomrule
\end{tabular}
\end{center} 

\section{Bounds on the VC-dimension}\label{second}

We extend \Cref{con} by considering families of $k$-connected subgraphs, for $k \geq 2$. Concerning the upper bound, the idea is to construct a spanning tree with at least $\textsc{VC}_{k-\mbox{\scriptsize{\normalfont{con}}}}(G) + k - 1$ leaves. We fix a shattered set $A$ of maximum cardinality and choose an appropriate vertex $r \in A$ as the root. Then we consider some $k$ neighbours of $r$, say $u_{1}, \dots, u_{k}$, and we try to ``attach'' the remaining vertices in $A$ to the graph $(\left\{r, u_{1}, \dots, u_{k}\right\}, \left\{ru_{1}, \dots, ru_{k}\right\})$ via appropriate paths.  

\begin{theorem}\label{kcon} $\textsc{VC}_{k-\mbox{\scriptsize{\normalfont{con}}}}(G) \leq \ell(G) - k + 1$, for any connected graph $G$ and $k \geq 2$.
\end{theorem}

\begin{proof} Let $A$ be a shattered set of maximum cardinality. Since our aim is to construct a spanning tree with at least $\left|A\right| + k - 1$ leaves, it would be enough to construct a tree $T \subseteq G$ with as many leaves. For $w \in A$, we denote by $G_{w}$ a fixed $k$-connected subgraph such that $V(G_{w}) \cap A = \{w\}$. Similarly, $G_{ww'}$ denotes a fixed $k$-connected subgraph such that $V(G_{ww'}) \cap A = \{w, w'\}$. We choose a vertex $r \in A$ having the minimum number of neighbours in $V(G) \setminus A$ as the root of $T$. Clearly, $d_{G_{r}}(r) \geq k$. Let $U = \{u_{1}, \dots, u_{k}\}$ be a set of $k$ arbitrary vertices in $N_{G_{r}}(r)$. We select the edges $u_{1}r, \dots, u_{k}r$. By Menger's Theorem \cite{West}, there exist $k$ independent $w-r$ paths in $G_{wr}$. We call $w \in A \setminus \{r\}$ a \textit{lower leaf} for $T$ if there exist $k$ independent $w-r$ paths $\left\{P_{1}, \dots, P_{k}\right\} \subseteq G$ with no inner vertex in $A$ and such that each of them contains (exactly) one edge in $\left\{u_{1}r, \dots, u_{k}r\right\}$ (in other words, $w$ is a lower leaf if there exists a $w, U$-fan in $G - (A \setminus \{w\})$). In particular, no path contains two vertices in $U$. Otherwise, we call $w$ an \textit{upper leaf} for $T$. 

We set $L := \left\{u_{1}, \dots, u_{k}\right\}$ and we view $L$ as the set of potential leaves for $T$. For any $w \in A \setminus \{r\}$, we do the following (see \Cref{fig:uplow}): 
\begin{itemize}
    \item If $w$ is an upper leaf, select a $w-r$ path $P \subseteq G_{wr}$ such that $V(P) \cap U = \varnothing$. Such a path exists: by Menger's Theorem \cite{West}, there exist $k$ independent $w-r$ paths in $G_{wr}$ and, if each of them contained a (different) vertex in $U$, we would obtain a $w, U$-fan in $G - (A \setminus \{w\})$. Finally, add $w$ to $L$ and remove cycles and appropriate edges so that the selected subgraph is a tree;
    \item If $w$ is a lower leaf, select a $w-r$ path $P \subseteq G$ as in the definition of lower leaf and such that $E(P) \cap \left\{u_{1}r, \dots, u_{k}r\right\} = \{u_{1}r\}$, add $w$ to $L$ and remove $u_{1}$ from $L$. Finally, remove edges from the newly added paths so that the selected subgraph is a tree.
\end{itemize}

\begin{figure}[h!]
\centering
\includegraphics[scale=0.4]{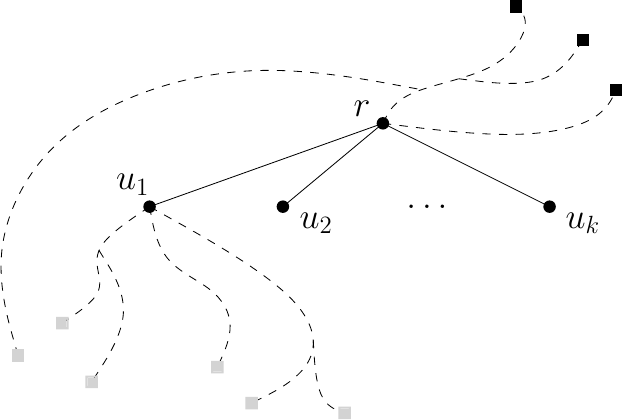}
\caption{The black square vertices are the upper leaves, while the gray square vertices are the lower leaves. The selected paths are dashed.}
\label{fig:uplow}
\end{figure}

In this way, we get a tree $T$ rooted at $r$ and in which the elements of $L$ are leaves. Moreover, $|L| \geq \left|A\right| + k - 2$. The construction works for any $k$, the case $k = 1$ giving the upper bound in \Cref{con}. Now we consider $k \geq 2$. The following list of claims gives some properties a counterexample to \Cref{kcon} would have.  

\begin{claim}\label{adj}
Each leaf in $T$ is adjacent to at most one vertex not in $T$. 
\end{claim}

Indeed, if there exists a leaf in $T$ adjacent to at least two vertices not in $T$, we immediately get a tree with at least $\left|A\right| + k - 1$. \openbox

\begin{claim}\label{upperleaf} $T$ contains at least one upper leaf.
\end{claim}

Indeed, suppose $T$ has no upper leaves. For any $2 \leq i \leq k$, select a $u_{i}-u_{1}$ path in $G_{r}$ with no inner vertex in $(U \setminus \{u_{1}, u_{i}\}) \cup \{r\}$. Clearly, such a path has no inner vertex in $A \cup U$. But then we can obtain a tree, rooted at $u_{1}$, with at least $\left|A\right| + k - 1$ leaves ($r$ becomes an additional leaf). \openbox

\begin{claim}\label{length2} There exists an upper leaf $w$ with $d_{T}(w, r) \geq 2$.
\end{claim}

Indeed, suppose this is not the case. By \Cref{upperleaf}, the set $W$ of upper leaves for $T$ is non-empty and each of them is adjacent to $r$. Suppose that any $w \in W$ has one neighbour (in $G$) which is contained in $T_{u_{1}} - A$, where $T_{u_{1}}$ denotes the subtree induced by $u_{1}$ and its descendants. Then we select the edges joining each upper leaf to $T_{u_{1}} - A$ and, for any $2 \leq i \leq k$, we select a $u_{i}-u_{1}$ path in $G_{r}$ with no inner vertex in $A \cup U$. In this way, we obtain a tree rooted at $u_{1}$ and with at least $|A| + k - 1$ leaves, a contradiction. Therefore, there exists $w \in W$ such that $N_{G}(w) \cap (V(T_{u_{1}}) \setminus A) = \varnothing$. By \Cref{adj}, $w$ has at most one neighbour in $V(G) \setminus V(T)$. Therefore, $N_{G_{w}}(w) \subseteq N_{G}(w) \setminus A \subseteq \{u_{2}, \dots, u_{k}\} \cup \{x\}$, for some $x \in V(G) \setminus V(T)$. But since $w$ has at least $k$ neighbours in $G_{w}$, we have $N_{G_{w}}(w) = N_{G}(w) \setminus A = \{u_{2}, \dots, u_{k}\} \cup \{x\}$. 

Consider now $w' \in A$ such that $ww' \notin E(G)$. There exists a $w-w'$ path $P$ in $G_{ww'}$ with no inner vertex in $U \setminus \{u_{1}\}$. Clearly, $wx \in E(P)$. Note that $P$ does not contain any vertex in $T_{u_{1}} - A$, or else $x$ could become an additional leaf for $T$. By selecting these paths, together with edges connecting vertices in $A$ to $w$, it is easy to see we can get a new tree $T'$ rooted at $w$ and with at least $|A| + k - 2$ leaves. Moreover, if an upper leaf for $T$ has at least two neighbours in the subtree $T_{u_{1}} - A$, then we can get an additional leaf for $T'$. Therefore, each upper leaf for $T$ has at most one neighbour in $T_{u_{1}} - A$. 

Now let $w' \neq w$ be an upper leaf for $T$. We claim that $w'$ is adjacent to $u_{2}$. Indeed, suppose $w'u_{2} \notin E(G)$. By \Cref{adj}, $w'$ has at most one neighbour in $G - V(T)$ and by the paragraph above it has at most one neighbour in $T_{u_{1}} - A$. Since $d_{G_{w'}}(w') \geq k$, then $N_{G_{w'}}(w') = (U \setminus \{u_{1}, u_{2}\}) \cup \{y, z\}$, for some $y \in V(G) \setminus V(T)$ and $z \in V(T_{u_{1}}) \setminus A$. Moreover, $u_{2} \notin V(G_{w'})$. Otherwise, there exists a $w'-u_{2}$ path $P$ in $G_{w'}$ with no inner vertex in $(U \setminus \{u_{1}, u_{2}\}) \cup \{z\}$. Clearly, $w'y \in E(P)$ and $P$ does not contain vertices in $T_{u_{1}} - A$, or else $y$ would become an additional leaf for $T$. But then there exists a $w', U$-fan in $G - (A \setminus \{w'\})$, contradicting the fact that $w'$ is an upper leaf. Therefore, since $u_{2} \notin V(G_{w'})$, there exists a $y-z$ path in $G_{w'}$ with no inner vertex in $A \cup (U \setminus \{u_{1}\})$. Again, by adding this path to the initial tree $T$, we get a new tree with at least $|A| + k - 1$ leaves, a contradiction. This means that $w'u_{2} \in E(G)$, for any upper leaf $w'$. On the other hand, for any lower leaf $w''$, there exists a $w''-u_{2}$ path with no inner vertex in $A \cup U$. Finally, for any $u_{i} \in U \setminus \{u_{2}\}$, there exists a $u_{i}-u_{2}$ path in $G_{r}$ with no inner vertex in $A \cup U$. But then it is easy to construct a new tree rooted at $u_{2}$ and with at least $|A| + k - 1$ leaves. \openbox 

\begin{claim}\label{onepath} $r$ has at most one neighbour in $V(T) \setminus (A \cup U)$.
\end{claim}

Indeed, suppose $r$ has at least two neighbours in $V(T) \setminus (A \cup U)$. The subtree $T - (\{u_{2}, \dots, u_{k}\} \cup (A \setminus \{r\}))$ contains a leaf $q$ and let $Q = N_{T}(q) \cap (A \setminus \{r\})$. Let $w \in Q$ and consider $G_{w}$. Clearly, $w$ has at least $k - 1 \geq 1$ neighbours in $V(G) \setminus (A \cup \{q\})$. If for any $w \in Q$, one of these neighbours is contained in $V(T) \setminus (\{u_{2}, \dots, u_{k}\} \cup \{q\})$, then we get a tree rooted at $r$ and in which $q$ is an additional leaf. Therefore, there exists $w \in Q$ with no neighbours in $V(T) \setminus (\{u_{2}, \dots, u_{k}\} \cup \{q\})$. By \Cref{adj}, $w$ has at most one neighbour in $V(G) \setminus V(T)$. But then $w$ has at most $k + 1$ neighbours in $V(G) \setminus A$, contradicting the minimality of $r$. \openbox \\

By \Cref{onepath} and \Cref{length2}, we may assume $r$ has exactly $k + 1$ neighbours in $V(T) \setminus A$. Let $N_{T}(r) \setminus (A \cup U) = \{z\}$ and let $T_{z}$ be the subtree induced by $z$ and its descendants. The following is immediate. 

\begin{fact} $V(T_{u_{1}}) \cap V(T_{z}) = \varnothing$.
\end{fact}

\begin{claim}\label{upperz} No leaf of $T_{z}$ is a lower leaf for $T$.
\end{claim}

Indeed, suppose there exists a lower leaf $w$ for $T$ which is a leaf of $T_{z}$. This means that the $w, u_{1}$-subpath $P$ in the definition of $w$ intersects $T_{z}$. Then, select the $u_{1}, x$-subpath of $P$, where $x$ is the first intersection of $P$ with $T_{z}$, when traversed from $u_{1}$. Moreover, for any $u_{i} \in U \setminus \{u_{1}\}$, select a $u_{i}-u_{1}$ path in $G_{r}$ with no inner vertex in $A \cup U$ and remove the edge set $\{u_{2}r, \dots, u_{k}r, rz\}$. After removing cycles and appropriate edges from the selected subgraph, we get a tree $T'$ rooted at $u_{1}$ and with at least $\left|A\right| + k - 2$ leaves. Finally, by \Cref{adj} and minimality of $r$, any upper leaf for $T$ adjacent to $r$ has a neighbour in $V(T') \setminus (U \setminus \{u_{1}\} \cup A)$ and so $r$ could become an additional leaf for $T'$. \openbox \\   

Clearly, $T_{z} - A$ contains a leaf $q$ and let $Q = N_{T}(q) \cap A$. By \Cref{upperz}, $Q$ is a set of upper leaves. By an argument similar to the one in the proof of \Cref{onepath}, there exists $w' \in Q$ such that $N_{G_{w'}}(w') \subseteq N_{G}(w') \setminus A \subseteq (U \setminus \{u_{1}\}) \cup \{q, x\}$, for some $x \in V(G) \setminus V(T)$, and $w'$ is adjacent to at least $k - 2$ vertices in $U \setminus \{u_{1}\}$ (see \Cref{fig:initree}). Moreover, $U \setminus \{u_{1}\} \subseteq V(G_{w'})$, or else $\{q, x\} \subseteq V(G_{w'})$ and there would exists an $x-q$ path in $G_{w'}$ with no inner vertex in $A \cup (U \setminus \{u_{1}\})$, thus turning $x$ into an additional leaf for $T$.

\begin{figure}[h!]
\centering
\hspace*{\fill}
\begin{subfigure}[b]{0.3\textwidth}
\centering
\includegraphics[scale=0.4]{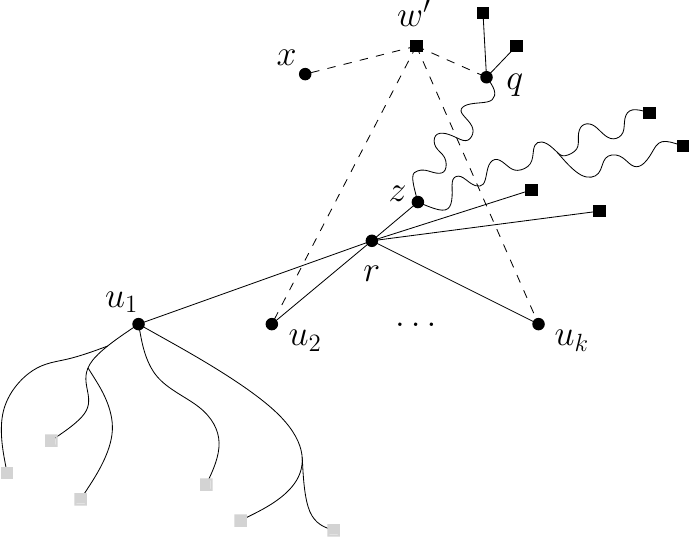}
\caption{}
\label{fig:initree}
\end{subfigure}
\hfill
\begin{subfigure}[b]{0.3\textwidth}
\centering
\includegraphics[scale=0.4]{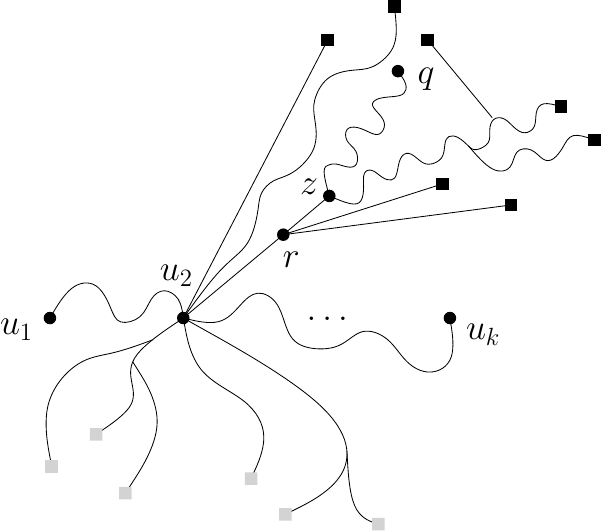}
\caption{}
\label{fig:endtree1}
\end{subfigure}
\hfill
\begin{subfigure}[b]{0.3\textwidth}
\centering
\includegraphics[scale=0.4]{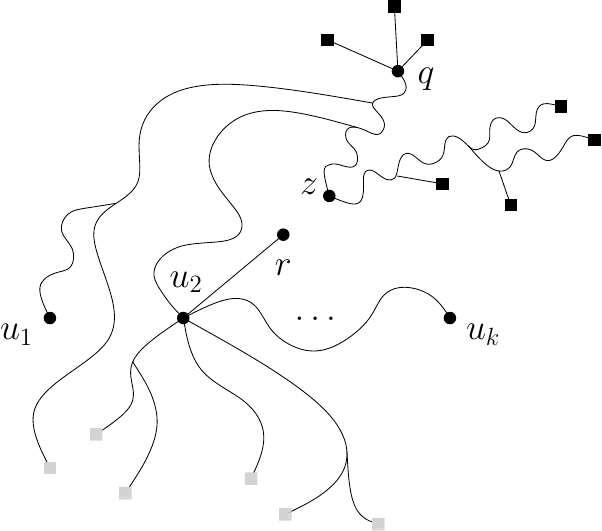}
\caption{}
\label{fig:endtree}
\end{subfigure}
\hspace*{\fill}
\caption{\subref{fig:initree} Neighbourhood of $w'$. \subref{fig:endtree1} and \subref{fig:endtree} The different constructions of a tree $T$ with at least $\left|A\right| + k - 1$ leaves.} 
\end{figure} 

\begin{claim}\label{path} There is no path $P$ from $w'$ to $v \in V(T_{u_{1}}) \setminus A$ with no inner vertex in $A \cup (U \setminus \{u_{1}\})$.
\end{claim} 

Suppose such a path $P$ exists. Then $w'x \notin E(P)$, or else $x$ becomes an additional leaf. Therefore, $w'q \in E(P)$. By the paragraph above, $w'$ is adjacent to at least $k - 2$ vertices in $U \setminus \{u_{1}\}$, say $\{w'u_{3}, w'u_{4}, \dots, w'u_{k}\} \subseteq E(G)$, and there exists a $w'-u_{2}$ path $P' \subseteq G_{w'}$ with no inner vertex in $(U \setminus \{u_{1}, u_{2}\}) \cup \{q\}$. If $w'u_{2} \in E(G)$, then there exists a $w', U$-fan in $G - (A \setminus \{w'\})$, contradicting the fact that $w'$ is an upper leaf. Therefore, $w'u_{2} \notin E(G)$ and $w'x \in E(P')$. But then $P + P(v)$, where $P(v)$ is the unique $v-u_{1}$ path in $T$, and $P'$ do not intersect in an inner vertex, or else $x$ becomes an additional leaf, and once again we obtain a $w', U$-fan in $G - (A \setminus \{w'\})$, a contradiction. \openbox  

\begin{claim}\label{oneneighb} Each $w \in A$ has at most one neighbour in $T_{u_{1}} - A$.
\end{claim}

Consider $w \in A \setminus \{w'\}$ such that $w'w \notin E(G)$. There exists a $w'-w$ path $P$ in $G_{w'w}$ with no inner vertex in $A \cup (U \setminus \{u_{1}\})$. Moreover, by \Cref{path}, $P$ does not contain any vertex of $T_{u_{1}} - A$. For any $w \in A \setminus \{w'\}$ such that $w'w \notin E(G)$, we select these paths. For any $w \in A \setminus \{w'\}$ such that $w'w \in E(G)$, we select the corresponding edges. Moreover, recall that $w'$ is adjacent to at least $k - 2$ vertices in $U \setminus \{u_{1}\}$, say $\{u_{3}, \dots, u_{k}\}$, and $U \setminus \{u_{1}\} \subseteq V(G_{w'})$. But then there exists a $w'-u_{2}$ path in $G_{w'}$ with no inner vertex in $A \cup U$ and so, by \Cref{path}, with no inner vertex in $T_{u_{1}} - A$ as well. Therefore, we can obtain a new tree $T'$ rooted at $w'$ and with at least $|A| + k - 2$ leaves. But then each $w \in A$ has at most one neighbour in $T_{u_{1}} - A$, or else we could get an additional leaf for $T'$. \openbox

\begin{fact}\label{fan} Let $w$ be an upper leaf. Then $((U \setminus \{u_{1}\}) \cup \{v\}) \nsubseteq V(G_{w})$, for any $v \in V(T_{u_{1}})$.
\end{fact}

Indeed, suppose this is not the case and let $v \in V(T_{u_{1}})$ be a vertex with minimum $d_{T}(v, u_{1})$ among those satisfying $((U \setminus \{u_{1}\}) \cup \{v\}) \subseteq V(G_{w})$. Since the set $U' = (U \setminus \{u_{1}\}) \cup \{v\}$ is contained in $V(G_{w})$ then, by the Fan Lemma \citep{West}, there exists a $w, U'$-fan in $G_{w}$. But then no path in the fan intersects the unique $v-u_{1}$ path in $T$ in an inner vertex and so we can obtain a $w, U$-fan in $G - (A \setminus \{w\})$, contradicting the fact that $w$ is an upper leaf. \openbox

\begin{claim}\label{empt} $V(G_{w}) \cap V(T_{u_{1}}) = \varnothing$, for any $w \in Q$.
\end{claim}

Indeed, suppose there exists $w \in Q$ such that $z \in V(G_{w}) \cap V(T_{u_{1}})$. By \Cref{fan}, we have $U \setminus \{u_{1}\} \nsubseteq V(G_{w})$. By \Cref{oneneighb}, $w$ has at most one neighbour in $T_{u_{1}} - A$ and, by \Cref{adj}, $w$ has at most one neighbour in $V(G) \setminus V(T)$. But then there exists $v \in V(G_{w}) \cap V(T_{z})$, or else $N_{G_{w}}(w) = (U \setminus \{u_{1}, u_{i}\}) \cup \{x, y\}$, for some $x \in V(G) \setminus V(T)$, $y \in V(T_{u_{1}}) \setminus A$ and $2 \leq i \leq k$, and we could find an $x-y$ path in $G_{w}$ with no inner vertex in $A \cup U$. Therefore, there exists a $v-z$ path $P \subseteq G_{w}$ with no inner vertex in $A \cup (U \setminus \{u_{1}\})$, contradicting \Cref{path}. \openbox \\

Consider now $w \in Q$. By \Cref{empt} and \Cref{adj}, if $V(G_{w}) \cap (V(T_{z}) \setminus \{q\}) = \varnothing$, then it must be $N_{G_{w}}(w) \subseteq (U \setminus \{u_{1}\}) \cup \{q, x\}$, for some $x \in V(G) \setminus V(T)$. Moreover, as we have seen before, $U \setminus \{u_{1}\} \subseteq V(G_{w})$, or else $x$ could become an additional leaf for $T$. Therefore, we proceed as follows. We start by removing $V(T_{u_{1}}) \cup U$ from $T$. 

Suppose now that, for every lower leaf $w''$, the $w'', u_{2}$-subpath $P$ in the definition of $w''$ contains no vertices in $T_{z}$ and, for every $u_{i} \in U \setminus \{u_{2}\}$, each $u_{i}-u_{2}$ path $P'$ in $G_{r}$ with no inner vertex in $A \cup U$ contains no vertices in $T_{z}$. In particular, these paths do not contain $q$ and we select them (see \Cref{fig:endtree1}). Moreover, for any $w \in Q$, we do the following. If there exists $v \in V(G_{w}) \cap (V(T_{z}) \setminus \{q\})$, then we select a $w-v$ path in $G_{w}$ with no inner vertex in $A \cup (U \setminus \{u_{2}\}) \cup \{q\}$ (such a path exists by \Cref{empt}). Otherwise, by the paragraph above, we select a $w-u_{2}$ path in $G_{w}$ with no inner vertex in $A \cup U \cup \{q\}$ (again, the existence follows from \Cref{empt}). After removing cycles and appropriate edges we get a new tree $T$, rooted at $u_{2}$ and with at least $|A| + k - 1$ leaves ($q$ becomes an additional leaf).  

Therefore, we may assume there exists either a lower leaf $\overline{w}$ such that the $\overline{w}, u_{2}$-subpath $P$ in the definition of $\overline{w}$ contains vertices in $T_{z}$, or a $u_{i} \in U \setminus \{u_{2}\}$ such that a $u_{i}-u_{2}$ path $P'$ in $G_{r}$ with no inner vertex in $A \cup U$ contains vertices in $T_{z}$. Then we select either the $\overline{w}, x$-subpath and the $y, u_{2}$-subpath of $P$, where $x$ and $y$ are, respectively, the first and the last intersection of $P$ with $T_{z}$ when traversed from $\overline{w}$, or the $u_{i}, x'$-subpath and the $y', u_{2}$-subpath of $P'$, where $x'$ and $y'$ are, respectively, the first and the last intersection of $P'$ with $T_{z}$ when traversed from $u_{i}$ (see \Cref{fig:endtree}). In this way, we get a new tree $T'$ and we grow it as follows. Let $w''$ be a lower leaf and $P_{w''}$ be the $w'', u_{2}$-subpath in the definition of $w''$. For any lower leaf $w''$, we select the $w'', x_{w''}$-subpath of $P_{w''}$, where $x_{w''}$ is the first intersection of $P_{w''}$ with the tree constructed so far and in which $w''$ becomes a leaf. Similarly, we add the remaining vertices of $U$ as leaves. Finally, consider an upper leaf $w''$ for the original $T$ such that $d_{T}(w'', r) = 1$. By minimality of $r$ and by an argument similar to \Cref{adj}, $w''$ is adjacent to a vertex in $V(T') \setminus (A \cup U \setminus \{u_{2}\})$. But then we can add these edges to the tree $T'$ rooted at $u_{2}$ in order to obtain a new tree with at least $|A| + k - 1$ leaves ($r$ becomes an additional leaf). 
\end{proof}

Our bound is tight in the sense of the following.

\begin{proposition}\label{tight} For any $k \geq 2$ and $x \geq 2k$, there exists a graph $G$ with $\ell(G) = x$ and $\textsc{VC}_{k-\mbox{\scriptsize{\normalfont{con}}}}(G) = \ell(G) - k + 1$.
\end{proposition}

For the proof we need the following elementary result which will be used in the upcoming sections as well.

\begin{lemma}[Expansion Lemma \citep{West}] If $G$ is a $k$-connected graph and $G'$ is obtained from $G$ by adding a new vertex with at least $k$ neighbours in $G$, then $G'$ is $k$-connected.
\end{lemma}

\begin{proof}[Proof of \Cref{tight}] For a fixed $k \geq 2$ and $x = 2k$, consider the graph $G_{k}$ constructed as follows. Start with a clique $H_{k}$ of size $k + 1$. For each subset $S \subset H_{k}$ of size $k$, add a vertex adjacent to precisely the vertices of $S$, and let $A$ be the set of the added vertices. Clearly, $\left|V(G_{k})\right| = 2k + 2$ and $\ell(G_{k}) = 2k$. Moreover, by the Expansion Lemma, $A$ is shattered. For $x \geq 2k$, let $G$ be the graph obtained from $G_{k}$ by adding $x - 2k$ vertices adjacent to exactly $k$ vertices of $H_{k}$. It is easy to see that $\ell(G) = x$ and $\textsc{VC}_{k-\mbox{\scriptsize{\normalfont{con}}}}(G) = \ell(G) - k + 1$.      
\end{proof}

As for a lower bound, we use the fact that having a sufficiently large complete subgraph guarantees shattering by $k$-connected subgraphs.

\begin{theorem}\label{lowerbound} Let $G$ be a connected graph of order $n$, size $m$, and maximum degree $\Delta$. For $k \geq 2$, $$\textsc{VC}_{k-\mbox{\scriptsize{\normalfont{con}}}}(G) \geq \ell(G) - k + 1 - \left(n + 2 - \left\lceil\frac{n - 2}{\Delta - 1}\right\rceil - \frac{n^{2}}{n^{2} - 2m}\right).$$
\end{theorem}

\begin{proof} By Tur\'{a}n's Theorem \cite{West}, if $m > \left(1 - \frac{1}{r}\right)\frac{n^{2}}{2}$, then $G$ contains $K_{r+1}$ as a subgraph. Therefore, by the Expansion Lemma, a set of size $r + 1 - (k + 1)$ can be shattered by $k$-connected subgraphs. The condition above is equivalent to $r < \frac{n^{2}}{n^{2} - 2m}$ and so, taking $r = \left\lceil\frac{n^{2}}{n^{2} - 2m} - 1\right\rceil$, we get $$\textsc{VC}_{k-\mbox{\scriptsize{\normalfont{con}}}}(G) \geq \left\lceil\frac{n^{2}}{n^{2} - 2m} - 1\right\rceil + 1 - (k + 1).$$ Let $T$ be a spanning tree of $G$ and $d_{i} = \left|\left\{v \in V(T) : d_{T}(v) = i\right\}\right|$. We want to find an upper bound for $\ell(G)$. We have that $$\sum^{\Delta}_{i=1}d_{i} = n \ \ \ \mbox{and} \ \ \ 2(n - 1) = \sum_{v \in V(T)}d_{T}(v) = \sum^{\Delta}_{i=1}id_{i}.$$ Using the two relations above, it is easy to see that $$n - d_{1} = \sum^{\Delta}_{i=2}d_{i} \geq \sum^{\Delta}_{i=2}\frac{i - 1}{\Delta - 1}d_{i} = \frac{1}{\Delta - 1}\sum^{\Delta}_{i=2}(i - 1)d_{i} = \frac{n - 2}{\Delta - 1}.$$ Since $\ell(G)$ is the maximum of $d_{1}$ taken over all spanning trees of $G$, then $$\ell(G) \leq n - \left\lceil\frac{n - 2}{\Delta - 1}\right\rceil.$$ Summarizing, we get 
\begin{align*}\textsc{VC}_{k-\mbox{\scriptsize{\normalfont{con}}}}(G) &\geq \left\lceil\frac{n^{2}}{n^{2} - 2m} - 1\right\rceil - k \\
                                   &\geq \frac{n^{2}}{n^{2} - 2m} - 1 - k + \left(\ell(G) - n + \left\lceil\frac{n - 2}{\Delta - 1}\right\rceil\right) \\
                                   &\geq \ell(G) - k + 1 - \left(n + 2 - \left\lceil\frac{n - 2}{\Delta - 1}\right\rceil - \frac{n^{2}}{n^{2} - 2m}\right).                                
\end{align*}
\end{proof}

Note that, in \Cref{lowerbound}, equality is attained by complete graphs.

\section{The decision problem}\label{third}

In this section we investigate the computational complexity of \textsc{Graph} $\textsc{VC}_{k-\mbox{\scriptsize{\normalfont{con}}}}$ \textsc{Dimension}. Consider the following decision problem, usually called \textsc{Set Multicover}:

\begin{center}
\fbox{%
\begin{minipage}{5.5in}
\textsc{Set Multicover}
\begin{description}[\compact\breaklabel\setleftmargin{70pt}]
\item[Instance:] A set $S = \left\{a_{1}, \dots, a_{n}\right\}$, a collection of subsets $S_{1}, \dots, S_{m} \subseteq S$, and integers $k$ and $t$.
\item[Question:] Is there an index set $I \subseteq \left\{1, \dots, m\right\}$ such that $\bigcup_{i \in I}S_{i} = S$, each $a_{i}$ is covered by at least $k$ distinct subsets, and $\left|I\right| \leq t$?  
\end{description}
\end{minipage}}
\end{center}	

Being a generalization of the well-known \textsc{Set Cover} (also known as \textsc{Minimum Cover} \cite{GJ79}), it is $\mathsf{NP}$-complete. We use it in the proof of the following \nameCref{reduction}. Recall that a \textit{split graph} is a graph in which the vertex set can be partitioned into a clique and an independent set. 

\begin{theorem}\label{reduction} \textsc{Graph} $\textsc{VC}_{k-\mbox{\scriptsize{\normalfont{con}}}}$ \textsc{Dimension} is $\mathsf{NP}$-complete even for split graphs.
\end{theorem}

\begin{proof} First we show that the problem is in $\mathsf{NP}$. Our proof is based on the following elementary \nameCref{union}. Since we could not find it in the literature, we give its short proof.

\begin{lemma}\label{union} Let $G$ and $G'$ be two $k$-connected graphs such that $\left|V(G) \cap V(G')\right| \geq k$. Then $G \cup G'$ is $k$-connected as well.
\end{lemma}

\begin{proof}[Proof of \Cref{union}] Let $S \subset V(G \cup G')$ be a subset such that $\left|S\right| < k$. Let $v$ and $w$ be two distinct vertices in $(G \cup G') - S$. If both $v$ and $w$ are in $G$ or in $G'$, then there is a $v-w$ path in $(G \cup G') - S$ by assumption. Otherwise, since $\left|V(G) \cap V(G')\right| \geq k$, there exists $u \in V(G) \cap V(G') \cap V((G \cup G') - S)$. Moreover, since $G - S$ and $G' - S$ are connected, there exist a $v-u$ path in $G - S$ and a $u-w$ path in $G' - S$. But then there is a $v-w$ walk in $(G \cup G') - S$ and so a $v-w$ path as well.   
\end{proof}

Let $G = (V, E)$ and $s \geq 1$ be an instance of \textsc{Graph} $\textsc{VC}_{k-\mbox{\scriptsize{\normalfont{con}}}}$ \textsc{Dimension}. We claim we can check in polynomial time whether a subset $V' \subseteq V$ with $\left|V'\right| \geq s$ is shattered. Indeed, by \Cref{union}, it is enough to check all the $\mathcal{O}(\left|V\right|^{k+1})$ subsets of $V'$ of size at most $k + 1$. Recall that finding a minimum separating set of a graph $G$ is polynomial in the order of $G$. Moreover, if $S \subseteq V(G)$ is a minimum separating set and $A \cup B = V(G - S)$ is a partition into two non-empty sets such that any path from a vertex in $A$ to a vertex in $B$ contains a vertex in $S$ then, for $k > |S|$, the vertices of any $k$-connected subgraph of $G$ are entirely contained in either $A \cup S$ or $B \cup S$. These observations, as shown in \citep[Theorem~1]{KSS93}, allow to test whether $G$ has a $k$-connected subgraph in polynomial time. Therefore, for any $B \subseteq V'$ of size at most $k + 1$, we can check in polynomial time if there exists a $k$-connected subgraph contained in $G - (V' \setminus B)$ and containing $B$. 

Now we prove $\mathsf{NP}$-hardness by a reduction from \textsc{Set Multicover}. Given an instance of \textsc{Set Multicover}, we construct a graph $G = (V, E)$ as follows (see \Cref{fig:multigraph}). The set of vertices $V$ is formed by four pairwise disjoint sets $A$, $B$, $C$ and $D$. $A$ is an independent set of $n \cdot (t + k + 1)$ vertices arranged in $n$ columns of $t + k + 1$ vertices each (every element in the $j$-th column corresponds to a copy of $a_{j}$), $B = \left\{v_{1}, \dots, v_{m}\right\}$ is a clique ($v_{i}$ corresponds to the set $S_{i}$), $C$ is a clique of size $k$ and $D$ is an independent set of $t + m + 1$ vertices. Each vertex in $C$ is connected to all vertices in $B$ (therefore, $B \cup C$ is a clique of size $m + k$) and $D$. Finally, $v_{i} \in B$ is connected to every copy of $a_{j} \in A$ if and only if $a_{j} \in S_{i}$. 

\begin{figure}[h!]
\centering
\includegraphics[scale=0.4]{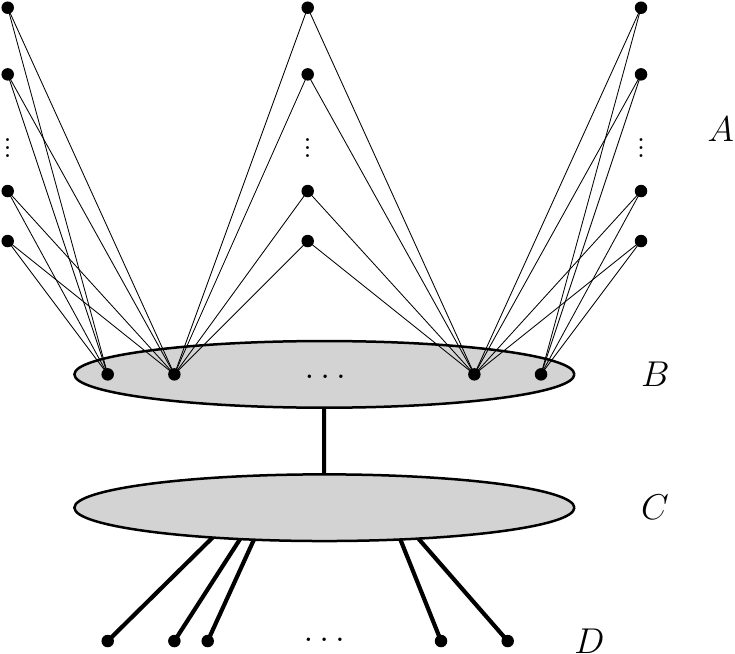}
\caption{The graph $G$ for the reduction. The gray ovals are cliques. A thick edge joining a vertex $v \in D$ to $C$ means that $v$ is adjacent to all the vertices of $C$. Similarly, the thick edge between the ovals means that $B \cup C$ is a clique.}
\label{fig:multigraph}
\end{figure}

Since $B \cup C$ is a clique and $A \cup D$ is an independent set, then $G$ is a split graph. We claim that there is an index set $I \subseteq \left\{1, \dots, m\right\}$ such that $\bigcup_{i \in I}S_{i} = S$, each $a_{i}$ is covered by at least $k$ distinct subsets and $\left|I\right| \leq t$ if and only if $\textsc{VC}_{k-\mbox{\scriptsize{\normalfont{con}}}}(G) \geq \left|V\right| - (t + k)$. 

Suppose first such an index set $I$ exists. We claim that the set $$V' = A \cup D \cup \left\{v_{i} \in B : i \notin I\right\}$$ is shattered. Indeed, the subgraph $G' = G[C \cup \{v_{i} \in B : i \in I\}]$ is a clique of size at least $k + 1$ and each vertex in $V'$ has at least $k$ neighbours in $G'$. Therefore, $V'$ is shattered by the Expansion Lemma. Finally, $\left|I\right| \leq t$ implies that $\left|V'\right| \geq \left|V\right| - t - k$.

Conversely, let $V'$ be a shattered set of maximum cardinality. Then $\left|V'\right| \geq \left|V\right| - (t + k)$. Suppose there exists $c \in V' \cap C$. Then no vertex in $D$ can be shattered, and so $\left|V'\right| \leq \left|V\right| - (t + m + 1) < \left|V\right| - (t + k)$. Therefore, no vertex of $C$ is in $V'$ and $D \subseteq V'$. Moreover, at least one vertex $v \in A$ for each column is in $V'$ and so $v$ has at least $k$ neighbours in the clique $B \setminus V'$. By the Expansion Lemma and since all the vertices in the column of $v$ have identical neighbourhoods, then each vertex in the column of $v$ belongs to $V'$ and so $A \subseteq V'$. Therefore, the number of vertices in $B$ which are in $V'$ is at least $\left|V\right| - (t + k) - \left|A\right| - \left|D\right| = m - t$. We claim that $I = \left\{i : v_{i} \in B \setminus V'\right\}$ is a yes-instance of \textsc{Set Multicover}. Indeed, any vertex of $A$ has at least $k$ neighbours in $B \setminus V'$. In other words, each $a_{j} \in S$ is contained in at least $k$ of the subsets $S_{i}$ with $i \in I$. Moreover, $\left|I\right| \leq m - (m - t) = t$.  
\end{proof}

\section{Graphs of bounded clique-width}

Graphs of bounded tree-width are particularly interesting from an algorithmic point of view: a lot of $\mathsf{NP}$-complete problems can be solved in linear time for them. For example, all graph properties which are expressible in monadic second-order logic with edge-set quantification are decidable in linear time on graphs of bounded tree-width. This is the content of the celebrated algorithmic meta-theorem of \citet{Cou90}. Let us briefly recall that monadic second-order logic is an extension of first-order logic by quantification over sets. The language of \textit{monadic second-order logic of graphs} ($\mbox{MSO}_{1}$ in short) contains the expressions built from the following elements:
\begin{itemize}
\item variables $x, y, \dots$ for vertices and $X, Y, \dots$ for sets of vertices;
\item predicates $x \in X$ and $\textit{adj}(x, y)$;
\item equality for variables, standard Boolean connectives and the quantifiers $\forall$ and $\exists$.
\end{itemize}
By considering the two-sorted incidence graph model, in which the edges form another sort of elements and we have the incidence predicate $\textit{inc}(v, e)$, we obtain \textit{monadic second-order logic of graphs with edge-set quantification} ($\mbox{MSO}_{2}$ in short).

If a graph property is expressible in the more restricted $\mbox{MSO}_{1}$, then \citet{CMR00} showed that it is decidable in linear time even on graphs of bounded clique-width, assuming a $k$-expression of the graph is explicitly given. The following observation shows that, for any graph, being shattered by its $k$-connected subgraphs is a property that makes no exception to this framework.

\begin{lemma} Being shattered by $k$-connected subgraphs is expressible in $\textsc{MSO}_{1}$.
\end{lemma}

\begin{proof} Let $G = (V, E)$ be a graph. The following $\mbox{MSO}_{2}$ sentence says that the subgraph induced by $X \subseteq V$ is connected: $$\textbf{conn}(X) = \forall_{Y \subseteq V}[(\exists_{u \in X} \ u \in Y \wedge \exists_{v\in X} \ v \notin Y) \rightarrow (\exists_{e \in E} \exists_{u \in X} \exists_{v \in X} \ \textit{inc}(u, e) \wedge \textit{inc}(v, e) \wedge u \in Y \wedge v \notin Y)].$$

It is easy to see that the quantification over single edges can be expressed by a $\mbox{MSO}_{1}$ sentence as follows: $$\exists_{a \in V} \exists_{b \in V} \exists_{u \in X} \exists_{v \in X} \ \textit{adj}(a, b) \wedge (u = a \vee u = b) \wedge (v = a \vee v = b) \wedge u \in Y \wedge v \notin Y.$$

The following $\mbox{MSO}_{1}$ sentence says that the subgraph induced by $X \subseteq V$ is $k$-connected: $$\textbf{k-conn}(X) = \exists_{v_{1} \in X} \dots \exists_{v_{k + 1} \in X}(\forall_{u_{1} \in V}\dots\forall_{u_{k - 1} \in V} \ \textbf{conn}(X \setminus \{u_{1}, \dots, u_{k - 1}\})).$$

Finally, the following $\mbox{MSO}_{1}$ sentence says that the set $A \subseteq V$ is shattered by $k$-connected subgraphs:
\begin{equation*} \textbf{shatt}(A) = \forall_{B \subseteq A}\exists_{X \subseteq V} \ \textbf{k-conn}(X) \wedge X \cap A = B.
\end{equation*}
\end{proof}

Therefore, as an immediate consequence of the meta-theorem stated above, we have the following.

\begin{corollary}\label{treeclique} \textsc{Graph} $\textsc{VC}_{k-\mbox{\scriptsize{\normalfont{con}}}}$ \textsc{Dimension} is decidable in linear time for graphs of bounded clique-width.
\end{corollary}

As shown by \citet{CO00}, every graph of bounded tree-width has bounded clique-width but there are graphs of bounded clique-width having unbounded treewidth (for example, complete graphs). Therefore, clique-width can be viewed as a more general concept than tree-width. Moreover, $P_{4}$-free graphs, also known as \textit{cographs}, are exactly those graphs having clique-width at most $2$ \citep{CO00}. See \citep{KLM09} for other classes of graphs of bounded clique-width.

We have seen that \textsc{Graph} $\textsc{VC}_{k-\mbox{\scriptsize{\normalfont{con}}}}$ \textsc{Dimension} is $\mathsf{NP}$-hard even for split graphs. Are there some subclasses of split graphs on which the problem becomes easy? Recall that the \textit{Dilworth number} of a graph $G$ is the size of a largest antichain (or, equivalently, the size of a minimum chain partition) with respect to the quasi-order $\preceq$ defined on the vertices of $G$ as follows: $x \preceq y$ if and only if $N_{G}(x) \subseteq N_{G}[y]$. Graphs with Dilworth number $1$ are precisely the well-known \textit{threshold} graphs, which are the $P_{4}$-free split graphs. Therefore, they have clique-width at most $2$ and we have seen we can decide the VC-dimension in linear time. On the other hand, already a small jump for the Dilworth number of a split graph, from $1$ to $2$, changes the clique-width from bounded to unbounded \citep{KLM14}. Nevertheless, deciding the VC-dimension remains easy.   

\begin{theorem}\label{thres} \textsc{Graph} $\textsc{VC}_{k-\mbox{\scriptsize{\normalfont{con}}}}$ \textsc{Dimension} is decidable in polynomial time for split graphs with Dilworth number at most $2$.
\end{theorem}

\begin{proof} Let $G = (V, E)$ be the input graph. We assume the unique partition of $V$ into a clique of size $\omega(G)$ and an independent set $I$ of size $\alpha(G)$ is given. It is well-known that the problem of finding a minimum chain partition of a poset can be translated into a maximum bipartite matching problem \citep{LP86}. Therefore, we can find in polynomial time a chain partition $I_{1} \cup I_{2}$ of $I$. For $j \in \{1, 2\}$, let $I_{j, \geq k} = \{u \in I_{j} : d(u) \geq k\}$ and $I_{j, \succeq u} = \{v \in I_{j} : v \succeq u\}$. Note that, if $\omega(G) \leq k$, then $G$ contains no $k$-connected subgraph. Therefore, we may assume $\omega(G) > k$. But then, a maximum size shattered set containing vertices from at most one between $I_{1}$ and $I_{2}$ has size $$\max \{\omega(G) - (k + 1) + |I_{1, \geq k}|, \ \omega(G) - (k + 1) + |I_{2, \geq k}|\}.$$ On the other hand, it is not difficult to see that, for any pair $x \in I_{1, \geq k}$ and $y \in I_{2, \geq k}$, a maximum size shattered set containing $x$ as the minimal element from $I_{1, \geq k}$ and $y$ as the minimal element from $I_{2, \geq k}$ has size
\begin{equation*}
\omega(G) - \max\{k + 1, \ 2k - |N(x) \cap N(y)|\} + |I_{1, \succeq x}| + |I_{2, \succeq y}|.
\end{equation*}
\end{proof}

It is not difficult to see that outerplanar graphs have tree-width at most $2$, and so bounded clique-width as well. On the other hand, planar graphs, in general, do not have bounded clique-width, even if restricted to be bipartite and with maximum degree $3$ (see \citep[Lemma~1]{KLM09}). In the next section we investigate the complexity of \textsc{Graph} $\textsc{VC}_{k-\mbox{\scriptsize{\normalfont{con}}}}$ \textsc{Dimension} when restricted to planar bipartite graphs with maximum degree at most $4$. 

\section{Subclasses of planar graphs}

For our reductions we need a series of auxiliary results. The following variant of \textsc{Hamiltonian Circuit Through Specified Edge} was shown to be $\mathsf{NP}$-complete in a comment on \citep{Lab11}.

\begin{center}
\fbox{%
\begin{minipage}{5.5in}
\textsc{Hamiltonian Circuit Through Specified Edge}
\begin{description}[\compact\breaklabel\setleftmargin{70pt}]
\item[Instance:] A planar cubic bipartite graph $G = (V, E)$ and $e \in E$.
\item[Question:] Does $G$ contain a Hamiltonian circuit through $e$?  
\end{description}
\end{minipage}}
\end{center}

For completeness, we give the details of the reduction from the following variant of \textsc{Hamiltonian Circuit}, shown to be $\mathsf{NP}$-complete by \citet{ANS80}.

\begin{center}
\fbox{%
\begin{minipage}{5.5in}
\textsc{Hamiltonian Circuit}
\begin{description}[\compact\breaklabel\setleftmargin{70pt}]
\item[Instance:] A planar cubic bipartite graph $G = (V, E)$.
\item[Question:] Does $G$ contain a Hamiltonian circuit?  
\end{description}
\end{minipage}}
\end{center}

\begin{theorem}[\citet{Lab11}] \textsc{Hamiltonian Circuit Through Specified Edge} is $\mathsf{NP}$-complete even for planar cubic bipartite graphs.
\end{theorem}

\begin{proof} Given a planar cubic bipartite graph $G = (V, E)$, we construct a graph $G' = (V', E')$ by replacing a vertex $u$ with the gadget depicted in \Cref{fig:edge} and we set $e = u_{2}'z$. Clearly, $G'$ is a planar cubic bipartite graph. 

\begin{figure}[h!]
\centering
\includegraphics[scale=0.5]{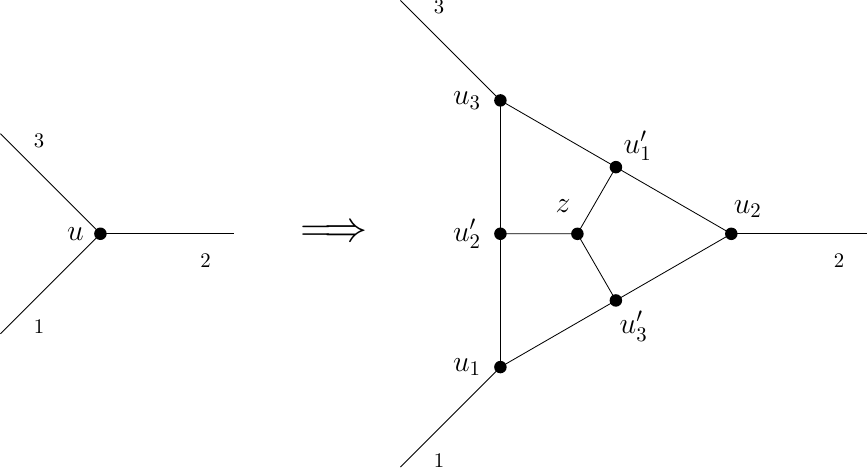}
\caption{Construction of the graph $G'$. The vertex $u$ is replaced by a gadget containing $7$ vertices.}
\label{fig:edge}
\end{figure}

It is easy to see that $G'$ contains a Hamiltonian circuit through $e$ if and only if $G$ contains a Hamiltonian circuit. 
\end{proof}

Now we strengthen some hardness results obtained by \citet{Dou92} and related to the problems \textsc{Hamiltonian Path} and \textsc{Connected Dominating Set}. We think they might be of independent interest. 

\begin{theorem}\label{ham} \textsc{Hamiltonian Path} is $\mathsf{NP}$-complete even for planar bipartite graphs where all the vertices have degree $3$, except two, which have degree $1$.  
\end{theorem}

\begin{proof} Let $G = (V, E)$ and $uv \in E$ be an instance of \textsc{Hamiltonian Circuit Through Specified Edge}, where $G$ is a planar cubic bipartite graph. Our reduction constructs a graph $G' = (V', E')$ as follows. Let $V' = V \cup \{a, b\}$, where $a$, $b$ are new vertices, and $E' = (E \setminus \{uv\}) \cup \{au, bv\}$. Clearly, $G'$ is a planar bipartite graph where all the vertices have degree $3$, except $a$ and $b$ of degree $1$. It is easy to see that $G$ has a Hamiltonian circuit through $uv$ if and only if $G'$ has a Hamiltonian path (between $a$ and $b$). 
\end{proof}

By \Cref{ham} and using the same reduction as in \citep[Theorem~1]{Dou92} (if the graph $G$ in \citep[Theorem~1]{Dou92} is bipartite, then $G'$ is bipartite too), we have the following. 

\begin{theorem} Given a planar bipartite graph $G = (V, E)$ with $\Delta(G) = 3$, it is $\mathsf{NP}$-complete to decide if there exists a spanning tree $T$ for $G$ such that $d_{T}(v)$ is either $1$ or $3$, for any $v \in V$.
\end{theorem}

Moreover, by the Proofs of \citep[Corollary~1 and Corollary~2]{Dou92}, we have that the following variant of \textsc{Connected Dominating Set} is $\mathsf{NP}$-complete.

\begin{center}
\fbox{%
\begin{minipage}{5.5in}
\textsc{Connected Dominating Set}
\begin{description}[\compact\breaklabel\setleftmargin{70pt}]
\item[Instance:] A planar bipartite graph $G = (V, E)$ with maximum degree $3$.
\item[Question:] Does there exist $D \subseteq V$ such that $G[D]$ is connected, every vertex in $V \setminus D$ is adjacent to at least one vertex in $D$ and $|D| \leq \frac{|V|}{2} - 1$?  
\end{description}
\end{minipage}}
\end{center}

Now we are ready to prove the $\mathsf{NP}$-hardness of \textsc{Graph} $\textsc{VC}_{\mbox{\scriptsize{\normalfont{con}}}}$ \textsc{Dimension} for the class of planar bipartite graphs with maximum degree $3$. Given the connection between $\textsc{VC}_{\mbox{\scriptsize{\normalfont{con}}}}$ and the connected domination number, hidden in \Cref{con}, it is no surprise we are going to reduce from the problem stated above.

\begin{theorem}\label{redpla} \textsc{Graph} $\textsc{VC}_{\mbox{\scriptsize{\normalfont{con}}}}$ \textsc{Dimension} is $\mathsf{NP}$-complete even for planar bipartite graphs with maximum degree $3$.
\end{theorem}

\begin{proof} We prove $\mathsf{NP}$-hardness by a reduction from the variant of \textsc{Connected Dominating Set} introduced above. Let $G = (V, E)$ be an instance of \textsc{Connected Dominating Set} where $G$ is a planar bipartite graph, $\Delta(G) = 3$ and $|V| = n$. We claim $G$ has a connected dominating set $D$ with $|D| \leq \frac{n}{2} - 1$ if and only if $\textsc{VC}_{\mbox{\scriptsize{\normalfont{con}}}}(G) \geq \frac{n}{2} + 1$. Clearly, we may assume $n \geq 46$. 

Suppose first $G$ has a connected dominating set $D$ with $|D| \leq \frac{n}{2} - 1$. Then each vertex in $V \setminus D$ can be joined to $G[D]$ independently of one another, and so $V \setminus D$ can be shattered by connected subgraphs. Therefore, $\textsc{VC}_{\mbox{\scriptsize{\normalfont{con}}}}(G) \geq |V \setminus D| \geq \frac{n}{2} + 1$.

Conversely, suppose $\textsc{VC}_{\mbox{\scriptsize{\normalfont{con}}}}(G) \geq \frac{n}{2} + 1$ and let $A$ be a shattered set of maximum size. By assumption, $|A| \geq 24$.

\begin{claim}\label{count} Each connected component $C$ of $G - A$ has at most $|C| + 2$ neighbours in $A$. 
\end{claim}

Indeed, if $C$ contains at most two $1$-vertices, then the claim clearly holds. Therefore, let $\{u_{1}, \dots, u_{k}\}$, with $k \geq 3$, be the set of $1$-vertices in $C$. Since $C$ is connected, there exists a $u_{1}-u_{2}$ path $P$ in $C$. Moreover, any $u_{3}-u_{1}$ path intersects $P$ in a $3$-vertex. Applying this reasoning again, we have that $d_{1}(C) - 2 \leq d_{3}(C)$, where $d_{i} = |\{v \in C : d_{C}(v) = i\}|$. But then $C$ has at most $2d_{1}(C) + d_{2}(C) = |C| + d_{1}(C) - d_{3}(C) \leq |C| + 2$ neighbours in $A$. \openbox \\
  
Since $\Delta(G) = 3$, then each vertex $u \in A$ has at least one neighbour in $G - A$, otherwise it would not be possible to shatter $u$ and a vertex in $A \setminus N(u)$. Let $C_{1}, \dots, C_{k}$ be the connected components of $G - A$. 

\begin{claim}\label{threecomp} There exists a vertex in $A$ joined to less than three connected components.
\end{claim}

Indeed, suppose each vertex in $A$ is joined to exactly three connected components. By \Cref{count} and double counting the size of the edge cut $[A, \overline{A}]$, we have $3|A| \leq \sum (|C_{i}| + 2)$. Therefore, $$k \geq \frac{3|A| - \sum |C_{i}|}{2} = \frac{3|A| - (n - |A|)}{2} \geq \frac{n}{2} + 2,$$ contradicting the fact that $|V \setminus A| \leq \frac{n}{2} - 1$. \openbox

\begin{claim}\label{twocomp} No vertex in $A$ is joined to exactly two connected components.
\end{claim}

Suppose, to the contrary, there exists $u \in A$ joined to exactly two connected components, say $C_{1}$ and $C_{2}$. Then, every vertex in $A \setminus N(u)$ is joined to either $C_{1}$ or $C_{2}$. By \Cref{count} we have $|C_{1}| + |C_{2}| + 3 \geq |A| - 1$, from which $|C_{1}| + |C_{2}| \geq \frac{n}{2} - 3$ and so $\sum^{k}_{i = 3}|C_{i}| \leq 2$. 

We claim that the (at least $|A| - 1$) vertices in $A \setminus N(u)$ are all joined to $C_{1}$ or all joined to $C_{2}$. Suppose, to the contrary, there exist $v$ and $w$ in $A \setminus N[u]$ such that $v$ is joined to $C_{1}$, $w$ is joined to $C_{2}$ but none of them is joined to both $C_{1}$ and $C_{2}$. Let $B \subseteq A$ be the subset of vertices which are joined to both $C_{1}$ and $C_{2}$. Since $|C_{1}| + 2 \geq |B|$ and $|C_{2}| + 2 \geq |B|$, then $|B| \leq \frac{|C_{1}| + |C_{2}|}{2} + 2 \leq \frac{n + 6}{4}$. Therefore, at least $\frac{n}{2} + 1 - \frac{n + 6}{4} = \frac{n - 2}{4}$ vertices of $A$ are not joined to both $C_{1}$ and $C_{2}$. But then the set $B'$ of vertices of $A \setminus (N[v] \cup N[w])$ which are not joined to both $C_{1}$ and $C_{2}$ has size at least $\frac{n - 2}{4} - 6 \geq 5$. Since $\{x, v\}$ and $\{x, w\}$ are shattered, for any $x \in B'$, then each vertex in $B'$ is joined to some connected component different from $C_{1}$ and $C_{2}$, contradicting the fact that the remaining connected components can be joined to at most four vertices in $A$. 

Therefore, the (at least $|A| - 1$) vertices in $A \setminus N(u)$ are all joined to the same connected component, say $C_{1}$. Suppose now there exists $u' \in N(u)$ not joined to $C_{1}$. By \Cref{count} we have $|C_{1}| \geq \frac{n}{2} - 2$. Moreover, the set $\{u', w\}$ is shattered, for any $w \in A \setminus \{u, u'\}$, contradicting the fact that the remaining component has at most $3$ neighbours in $A$. Therefore, each $v \in A$ is joined to $C_{1}$. Since $|C_{1}| \geq |A| - 2 \geq \frac{n}{2} - 1$, then $C_{2} = \varnothing$, a contradiction. \openbox \\

By \Cref{threecomp} and \Cref{twocomp}, there exists $v \in A$ joined to exactly one connected component $C$ of $G - A$. Then $v$ has at most two neighbours in $A$. Moreover, each of the (at least $|A| - 3$) nonneighbours of $v$ in $A$ is joined to $C$, otherwise it would not be possible to shatter the set $\{v, w\}$, for some $w \in A \setminus N[v]$. Suppose there exists $v' \in N(v) \cap A$ not joined to $C$. By \Cref{count} we have $|C| + 2 \geq |A| - 2$, from which $|C| \geq \frac{n}{2} - 3$ and so the remaining connected components contain at most two vertices. Since the set $\{v', w\}$ is shattered, for any $w \in A \setminus N[v']$, and there are at least $|A| - 3 \geq 21$ such nonneighbours, we get a contradiction to the fact that a component in $G - A - C$ can have at most four neighbours in $A$. Therefore, each $v \in A$ is joined to $C$ and $|C| \geq |A| - 2 \geq \frac{n}{2} - 1$. But then $|C| = \frac{n}{2} - 1$ and $C$ is a connected dominating set for $G$.  
\end{proof}

Clearly, \textsc{Graph} $\textsc{VC}_{\mbox{\scriptsize{\normalfont{con}}}}$ \textsc{Dimension} becomes easy for graphs $G$ with $\Delta(G) \leq 2$. Indeed, $\textsc{VC}_{\mbox{\scriptsize{\normalfont{con}}}}(P_{n}) = 2$, for $n \geq 3$, and $\textsc{VC}_{\mbox{\scriptsize{\normalfont{con}}}}(C_{3}) = 2$ and $\textsc{VC}_{\mbox{\scriptsize{\normalfont{con}}}}(C_{n}) = 3$, for $n \geq 4$. Note that, if $\Delta(G) \leq 2$, then $G$ has tree-width at most $2$. Therefore, the conclusion follows from \Cref{treeclique} as well.

We conclude by showing the $\mathsf{NP}$-hardness of \textsc{Graph} $\textsc{VC}_{2-\mbox{\scriptsize{\normalfont{con}}}}$ \textsc{Dimension} for planar bipartite graphs with $\Delta(G) = 4$. We leave as an open problem to determine what happens if we further impose $\Delta(G) = 3$.

\begin{theorem}\label{redpla2con} \textsc{Graph} $\textsc{VC}_{2-\mbox{\scriptsize{\normalfont{con}}}}$ \textsc{Dimension} is $\mathsf{NP}$-complete even for planar bipartite graphs with maximum degree $4$.
\end{theorem}

\begin{proof} Our reduction is from \textsc{Hamiltonian Circuit}, which remains $\mathsf{NP}$-complete even for planar cubic bipartite graphs \citep{ANS80}. Given an instance $G = (V, E)$ of \textsc{Hamiltonian Circuit}, where $G$ is a planar cubic bipartite graph with $|V| = n$, we construct a graph $G'$ by replacing each vertex $u$ of $G$ with the gadget depicted in \Cref{fig:hcgadget}. For $0 \leq i \leq 2$, the vertices $u_{i}$ are the \textit{gates} and the vertices $u_{i}'$ are the \textit{connectors}. Finally, each pair $u_{i}'u_{i + 1}'$ (indices modulo $3$) of connectors in the gadget is connected by a path of length $2$ with inner vertex $u_{i + 2}''$, called a \textit{crossing vertex}. Clearly, the construction can be done in polynomial time and the resulting graph $G' = (V', E')$ is planar, bipartite and $\Delta(G') = 4$. We claim that $\textsc{VC}_{2-\mbox{\scriptsize{\normalfont{con}}}}(G') \geq |V'| - 5n$ if and only if $G$ contains a Hamiltonian circuit. 

\begin{figure}[h!]
\centering
\includegraphics[scale=0.5]{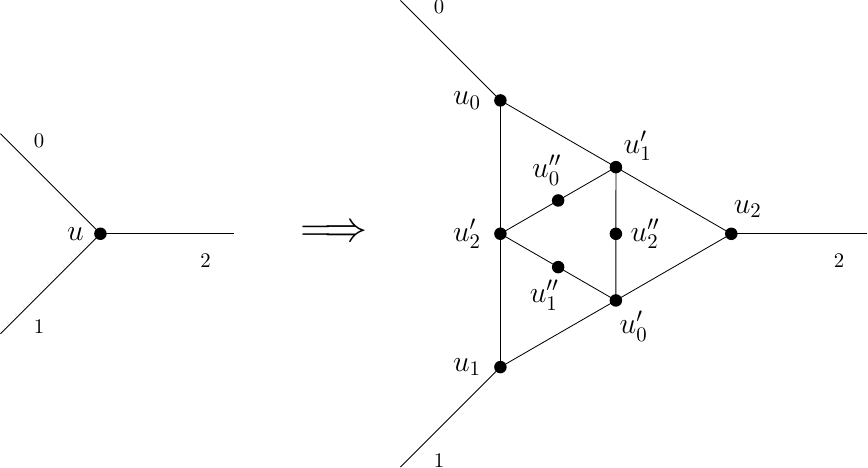}
\caption{Construction of the graph $G'$. The vertex $u$ is replaced by a gadget containing $9$ vertices.}
\label{fig:hcgadget}
\end{figure}

Suppose first $G$ contains a Hamiltonian circuit $C$. Suppose, without loss of generality, $u \in V$ is incident to the edges $1$ and $2$ in $C$. Then we augment the subgraph induced by $E(C)$ in $G'$ with the path $u_{1}u_{0}'u_{2}$. Repeating this procedure for every vertex of $G$, we get a circuit in $G'$ containing three vertices from every gadget. The vertices in $\{u_{0}, u_{0}'', u_{1}'', u_{2}''\}$ can be joined to this circuit, independently of one another, via paths through the connectors $u_{0}'$, $u_{1}'$ and $u_{2}'$. In all cases, the resulting subgraph is clearly $2$-connected. Repeating this process for every gadget, we have that a set of size $|V'| - 5n$ can be shattered.

Suppose now $\textsc{VC}_{2-\mbox{\scriptsize{\normalfont{con}}}}(G') \geq |V'| - 5n$ and let $A$ be a shattered set of maximum cardinality.

\begin{claim}\label{exact} For any gadget $H \subseteq G'$, exactly one gate and the three crossing vertices are in $A$.
\end{claim}

We show first that at most four vertices of $H$ are in $A$. Suppose, to the contrary, $H$ contains at least five vertices of $A$. Then at least one crossing vertex is in $A$, otherwise at least two gates would be in $A$, contradicting the fact that the set consisting of a connector in $H$ and a vertex not in $H$ is shattered. Therefore, at least one crossing vertex is in $A$, say without loss of generality $u_{1}''$. Then the connectors $u_{0}'$ and $u_{2}'$ are not in $A$. If another crossing vertex is in $A$, then $u_{1}' \notin A$ and at least two gates are in $A$, a contradiction. Therefore, $u_{1}''$ is the only crossing vertex in $A$. But then all the gates are in $A$, a contradiction again. 

Since $\textsc{VC}_{2-\mbox{\scriptsize{\normalfont{con}}}}(G') \geq |V'| - 5n$, then exactly four vertices per gadget are in $A$. We have seen that at most one gate per gadget is in $A$. Moreover, exactly one gate per gadget is in $A$, otherwise a crossing vertex and one of its neighbouring connectors would both be in $A$. Let $u_{0}$ be the gate of gadget $H$ in $A$. Suppose one of its neighbouring connectors is in $A$ (clearly, there exists at most one such connector). But then, again, a crossing vertex and at least one of its neighbouring connectors would both be in $A$, a contradiction. Therefore, both $u_{1}'$ and $u_{2}'$ are not in $A$ and it is easy to see that it must be $A \cap V(H) = \{u_{0}, u_{0}'', u_{1}'', u_{2}''\}$. \openbox \\

By \Cref{exact}, there exists a $2$-connected subgraph of $G'$ containing crossing vertices in every gadget and avoiding exactly one gate per gadget. Therefore, for any gadget, this subgraph contains exactly two of the edges incident to its gates and originally in $G$. This means that, contracting each gadget to a single vertex, we obtain a $2$-regular connected spanning subgraph. Therefore, $G$ contains a Hamiltonian circuit.   
\end{proof}

\section{Conclusion}

This work represents a continuation of the systematic study of the VC-dimensions of graphs initiated by \citet{KKRUW97}. We have concentrated on the VC-dimension with respect to $k$-connected subgraphs. In particular, we have proved the $\mathsf{NP}$-completeness of the associated decision problem and observed its decidability in linear time for graphs of bounded clique-width. In this context, we believe two interesting open problems arise: determine the complexity for unit interval graphs and for planar graphs in the remaining cases $3 \leq k \leq 5$. We conjecture that the first problem is in $\mathsf{P}$. Finally, it would be interesting to focus on the study of the VC-dimension with respect to other classes of subgraphs. 

\bibliographystyle{elsarticle-num-names} 
\bibliography{refs}

\end{document}